\documentclass[conference]{IEEEtran}
\IEEEoverridecommandlockouts
\usepackage{cite}
\usepackage{amsmath,amssymb,amsfonts}
\usepackage{algorithmic}
\usepackage{graphicx}
\usepackage{soul}
\usepackage{textcomp}
\usepackage{xcolor}
\usepackage{authblk}
\usepackage{amsthm}

\theoremstyle{plain}

 \newtheorem{prop}{Proposition}

 \newtheorem{lemma}{Lemma}

\def\BibTeX{{\rm B\kern-.05em{\sc i\kern-.025em b}\kern-.08em
    T\kern-.1667em\lower.7ex\hbox{E}\kern-.125emX}}
\begin{document}

\title{Embedding Jamming Attacks into Physical Layer Models in Optical Networks}
\author[1]{Mounir Bensalem} 
\author[2]{\'Italo Brasileiro} 
\author[2]{Andr\'e Drummond} 
\author[1]{Admela~Jukan} 
\affil[1]{Technische Universit\"at Braunschweig, Germany}
\affil[2]{ University of Brazilia, Brazil}


\maketitle

\begin{abstract}
Optical networks are prone to physical layer attacks, in particular the insertion of high jamming power. In this paper, we present a study of jamming attacks in elastic optical networks (EON) by embedding the jamming into the physical layer model, and we analyze its impact on the blocking probability and slots utilization. We evaluate our proposed model using a single link and a network topology and we show that for  in-band- jamming, the slots utilization decreases with the increase of jamming power, and becomes null when the jamming power is higher than 3 dB, while for out-of-band jamming, the impact is maximal for a specific jamming power, 1.75 dB in our simulation. Considering multiple positions of attackers, we attained the highest blocking probability 32\% for a specific jamming power 2 dB. We conclude that the impact of jamming depends on attacker positions as well as the jamming power. 
\end{abstract}
%

\section{Introduction} 

Elastic Optical Networks (EON) are highlighted as a possible infrastructure for large core networks, which make it critical to assure its security. Generally, optical networks are  vulnerable  to different  types  of  attacks,  including  power  jamming  attacks \cite{medard1998attack}. This physical attack consists on the insertion of high power signal over a specific frequency chosen by the attacker in order to decrease the signal-to-noise ratio (SNR). The jamming attacks are classified into two main types:  \textit{in-band jamming} and \textit{out-of-band jamming}. The in-band jamming considers the insertion of the high  power  signal over a frequency  within the  transmission  window to disrupt connections that use the jammed slots. The out-of-band jamming considers the insertion of the high  power  signal out  of  the transmission  band to disrupt connections that use the neighboring slots to the jammed ones. 

%

%
%

Previous research studied jamming attacks using experiments \cite{li2015fast, natalino2018field, jing2012experiment}  or simulations \cite{bensalem2019detecting, yuan2014protection}. Most works are focused on proposing solutions and models for detecting jamming attacks \cite{bensalem2019detecting, natalino2018field, li2018light} or protecting the network against its occurrence \cite{prucnal2009physical, hu2015chaos,li2016fast,singh2017combined}. However, before studying how to detect jamming attacks, or to design prevention mechanism, it is important to go back to the origin of the physical attack and analyze the non-linearity effects added by the insertion of a high power signal. An integration of jamming attacks model with the physical layer model is still missing in the literature. By simply having such integration to the physical layer model, we can derive a new analysis for jamming attacks, which does not exist in the current literature.\\

This work proposes a modeling for jamming attacks, in-band and out-of-band jamming, considering the embedding of the additional signal power into the physical layer model as non-linear impairments. A general physical layer model \cite{johannisson2014modeling} is used as reference for the power insertion, which causes a reduction in the Quality of Transmission (QoT) of the network circuits.\\
Our analysis shows that the proposed modeling reveals the non-linear characteristics of jamming attacks. To understand those characteristics, it is important to observe the impact of in-band and out-of-band jamming on slots utilization and the blocking probability. We studied two different network topologies, a single link and a small network of 6 nodes, considering different values of jamming power. We discussed as well the case of multiple attackers, where each one is using different frequency slots range. The results associated to in-band jamming,  showed that slots utilization decreases with the increase of jamming power, and becomes null when the jamming power is higher than 3 dB. On the other hand, the results observed for out-of-band jamming, which aims to harm neighboring circuits, shows that the impact is maximal for a specific jamming power, 1.75 dB in our simulation, and decreases when we use higher jamming power. Using multiple positions of attackers, we obtained the highest blocking probability 32\% for a specific jamming power 2 dB, as a result, we conclude that the impact of jamming depends on the jamming power as well as attacker positions.

The rest of this paper is organized as follows. Section \ref{sec:jamming} shows the proposed jamming model for EON. Section \ref{sec:evaluation} presents the performance evaluation for the proposed model. Finally, we conclude the paper in Section \ref{sec:conclusion}.

\section{Jamming Attack Model in EON}\label{sec:jamming}
The infrastructure is modeled as a graph $G=(V,E)$, where $V$ represents the set of nodes consisting of transceivers and routing devices, and $E$ the set of optical links. An optical link contains a set of connected fiber spans, such that each span is an optical fiber followed by an erbium-doped fiber amplifier (EDFA). Each link can transmit M slots. We define the center frequency $f_m$, the bandwidth $\Delta f_m$, and the power $P_m$, for each channel $m=1,...,M$.  The studied network can support an enormous number of connections established between different nodes and following a predefined route. A connection $i$ can be associated to a route $r_i$, where $r_i$ is a set of links starting from the source node and ending by the destination node of connection $i$.\\
There are two types of  jamming attacks:  \textit{In-band jamming} in which an attacker inserts a relatively high power signal over either a channel within the transmission window, i.e. the attacker can get access to an optical fiber, bend it, then insert a slightly high power  to the victim channel. \textit{Out-of-band jamming} where the high power signal is inserted out of the transmission band of legitimate data channels, i.e. the attacker is an unauthorized user being able to sneak into the network and access a light source (laser) to launch a signal with higher power than other authorized channels, targeting nonlinearity in the fiber and degradation of Optical Signal to Noise Ratio (OSNR) and Bit Error Rate (BER) of other neighboring channels.\\
Figure \ref{fig:arch1} presents an illustration of the studied jamming attack model in an optical  network, including a transmitter, series of spans (EDFA connected to an optical fiber) and a  receiver. The two type of jamming attack are shown in the figure. The out-of-band jamming is represented by an unauthorized signal inserted by Jammer 1 in the $2^{nd}$ channel with $\epsilon$ dB higher power than the legitimate  channels. The in-band jamming is represented by Jammer $2$ who gets an access to the fiber after multiplexing all signals. Jammer 2 choose a victim channel (ch4) and add $\epsilon$ dB power on top of the legitimate one.
\begin{prop}
The power of a jammed channel $P_J$ is equal to $P+\epsilon$  dB, in both type of attacks (in-band and out-of-band jamming). 
\end{prop}
\begin{proof}
it is clear.
\end{proof}
\begin{figure}
 \centering
   \includegraphics[scale=0.55]{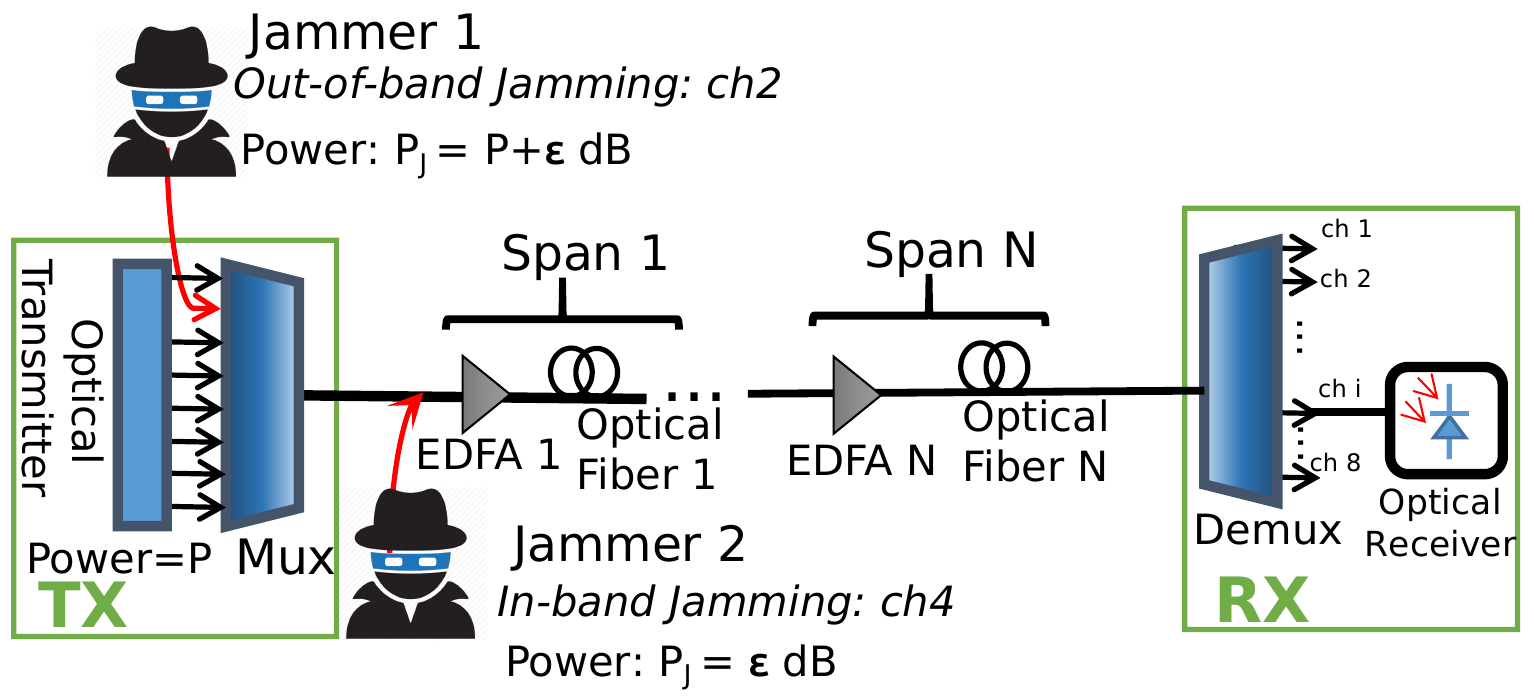}
 \caption{An Illustration of In-band and out-of-band Jamming Attack in EON}
\label{fig:arch1}
\vspace{-0.2 cm}
\end{figure}

Therefore, we formulate our model as follows: a jammer $J$ can insert a high power signal in one or more channels in order to disrupt the service. We denote by $M_J \subset M$ the set of jammed channels. We assume that the power of $J$ is $P_J= P + \epsilon$, where $P$ is the power of the rest of channels, and $\epsilon$ is the additional power. \\

The proposed jamming attack aware physical impairment model can be resumed by the following lemma:
\begin{lemma}
The jamming attack aware SNR equation associated to a connection $i$ using a route $r_i$ can be expressed as follows:

\begin{equation}\label{eq:snr}
    SNR_m = \frac{G_m}{G^{ASE} + G^{NLI,s} + G^{J}}
\end{equation}
Where $G_m$ is the signal power spectral density (PSD) of connection $i$ using a channel $m$, $G^{ASE}$ is the PSD of the amplified spontaneous emission (ASE) noise, $G^{NLI,s}$ is the PSD of the noise from nonlinear impairments (NLI) in a secure network, and $G^J$ is the PSD of the noise from nonlinear impairments (NLI) caused by jamming.
\end{lemma}
\begin{proof}

We consider the physical layer model in \cite{johannisson2014modeling}, which provides an explicit expression of the signal-to-noise ratio (SNR) considering the characteristics of the non-linear interference (NLI).  The SNR associated to a connection $i$ using a route $r_i$ can be expressed as follows:
\begin{equation}
    SNR_m = \frac{G_m}{G^{ASE} + G^{NLI}}
\end{equation}
Where  $G^{NLI}$ is the PSD of the noise from nonlinear impairments (NLI).\\
The PSD of connection $i$ using a channel $m$ is written as:
\begin{equation}\label{PSD}
    G_m = \frac{P_m}{\Delta_m}
\end{equation}
Where $P_m$ is the power corresponding to $G_m$, and $\Delta_m$ denotes the bandwidth and considered constant for all links in the connection.\\
The PSD of ASE can be expressed as follows:
\begin{equation}
    G^{ASE}=\sum_{l\in r_i}N_l G_{0}^{ASE},\;\;\;\; G_{0}^{ASE}=(e^{\alpha L} - 1)Fh\nu
\end{equation}
Where $N_l$ is the number of spans on link $l$, $L$ is the length of each span, $\alpha$ is the power attenuation, $F$ is the spontaneous emission factor, $h$ is Planck’s constant, and $\nu$ is the light frequency.\\
The PSD of the noise from NLI is:
\begin{equation}
    G^{NLI}=\sum_{l\in r_i}N_l G_{l}^{NLI}
\end{equation}
and using eq. (16) in \cite{johannisson2014modeling} and eq.(1) in \cite{yan2015link},
\begin{multline}
    G_{l}^{NLI}(f_m) = \phi G_m [G_m^2 \text{arcsinh}(\rho (\Delta f_m)^2) \\ + \sum_{\substack{m' = 1 \\ m' \neq m}}^{M}G_{m'}^2 \text{ln}(\frac{f_{m,m'} + \Delta f_{m'}/2}{f_{m,m'} - \Delta f_{m'}/2}) ]
\end{multline}
and,
\begin{equation}
    \phi = \frac{3\gamma^2}{2 \pi \alpha |\beta_2|}, \;\;\;\;\; \rho = \frac{\pi^2 |\beta_2|}{2\alpha}
\end{equation}
Where $f_m$ is the center frequency of channel $m$.  $m'$ is another connection using link $l$, $\Delta f_{m'}$ is the bandwidths for connection $m'$, $f_{m,m'}=|f_m - f_{m'}|$ is the center frequency spacing between connections $m$ and $m'$, $\gamma$ is the fiber nonlinearity coefficient, and $\beta_2$ is the fiber dispersion.\\

 Using  eq. (\ref{PSD}), the squared PSD of a jammed channel $m'$ can be expressed in terms of a non jammed channel $m$ as: 
\begin{equation}
    G_{m'}^2 = \frac{P_J^2}{\Delta_{m'}^2}=\frac{(P+\epsilon)^2}{\Delta_{m}^2} = G_m^2 + \frac{\epsilon^2 + 2\epsilon P}{\Delta_m^2}
\end{equation}
Thus, 
\begin{equation}
    G_{l}^{NLI}(f_m) = G_{l}^{NLI, s}(f_m)+ G_{l}^{J}(f_m)
\end{equation}
   such that,
    \begin{multline}
    G_{l}^{NLI, s}(f_m) = \phi G_m^3 [\text{arcsinh}(\rho (\Delta f_m)^2) \\ +  \sum_{\substack{m' = 1 \\ m' \neq m}}^{M}  \text{ln}(\frac{f_{m,m'} + \Delta f_{m'}/2}{f_{m,m'} - \Delta f_{m'}/2}) ]
\end{multline}
and,
\begin{multline}
    G_{l}^{J}(f_m) = \phi G_m  \frac{\epsilon^2 + 2\epsilon P}{\Delta_m^2} \sum_{\substack{m' \in M_J \\ m' \neq m}}\text{ln}(\frac{f_{m,m'} + \Delta f_{m'}/2}{f_{m,m'} - \Delta f_{m'}/2}) 
\end{multline}
Where $ G_{l}^{NLI, s}(f_m)$ is the PSD of the noise from NLI in a secure network in link $l$ (used in \cite{fontinele2017efficient} and \cite{zhao2015nonlinear}), and  $ G_{l}^{J}(f_m)$ is the PSD of the noise from NLI caused by jamming in link $l$. Thus, we define $G^{NLI,s}=\sum_{l\in r_i}N_l G_{l}^{NLI,s}$ the PSD of the noise from nonlinear impairments (NLI) in a secure network, and $G^J=\sum_{l\in r_i}N_l G_{l}^{J}$ the PSD of the noise from nonlinear impairments (NLI) caused by jamming. \\
Therefore, the jamming attack aware SNR equation associated to a connection $i$ using a route $r_i$ can be expressed as eq. (\ref{eq:snr}).

 
\end{proof}

In the network operational phase, the SNR value is verified by the control plane, as a way of maintaining the QoT of connections according to the standard provided by the Service Level Agreement (SLA). Therefore, to have an acceptable QoT, circuits should maintain the SNR value above the SNR threshold ($\text{SNR}_{th}$) allowed by the network. Thus, circuits with a SNR lower than the $\text{SNR}_{th}$ are not established. Moreover, circuits that will decrease  the SNR of the already established circuits to have a value lower than the  $\text{SNR}_{th}$  are blocked as well.

In the following Section,  we present the performance evaluation of the proposed model using link and network simulations.

\section{Performance Evaluation} \label{sec:evaluation}

In this section, we analyze the impact of jamming attacks on online network scenarios. We start our analysis using a basic network topology with two nodes and one bidirectional link between them. It is important to figure out the effect of jamming in one link to understand the nonlinearity patterns that may occur due to such high power attack. During the simulation, we choose to launch the jamming in a fixed slot range. The reason of fixing the position of the attacker is to be able to study the impact of in-band jamming on the chosen slots, and out-of-band jamming along the neighboring slots by looking to slots utilization. Afterward, we consider a more general network scenario including 6 nodes and 8 bidirectional links. 
In this topology, the occurrence of jamming is considered in only one link during the whole simulation, and similar to the link scenario, the jamming slot range is fixed. The effect of jamming can be propagated to other links through the circuits that cross the jammed slots in the attacked link. Figure \ref{fig:topologies} shows the topologies used in the simulations,  and the distance of each link in km.

\begin{figure}[!h]
    \begin{center}
        \includegraphics[width=0.5\textwidth, page=1]{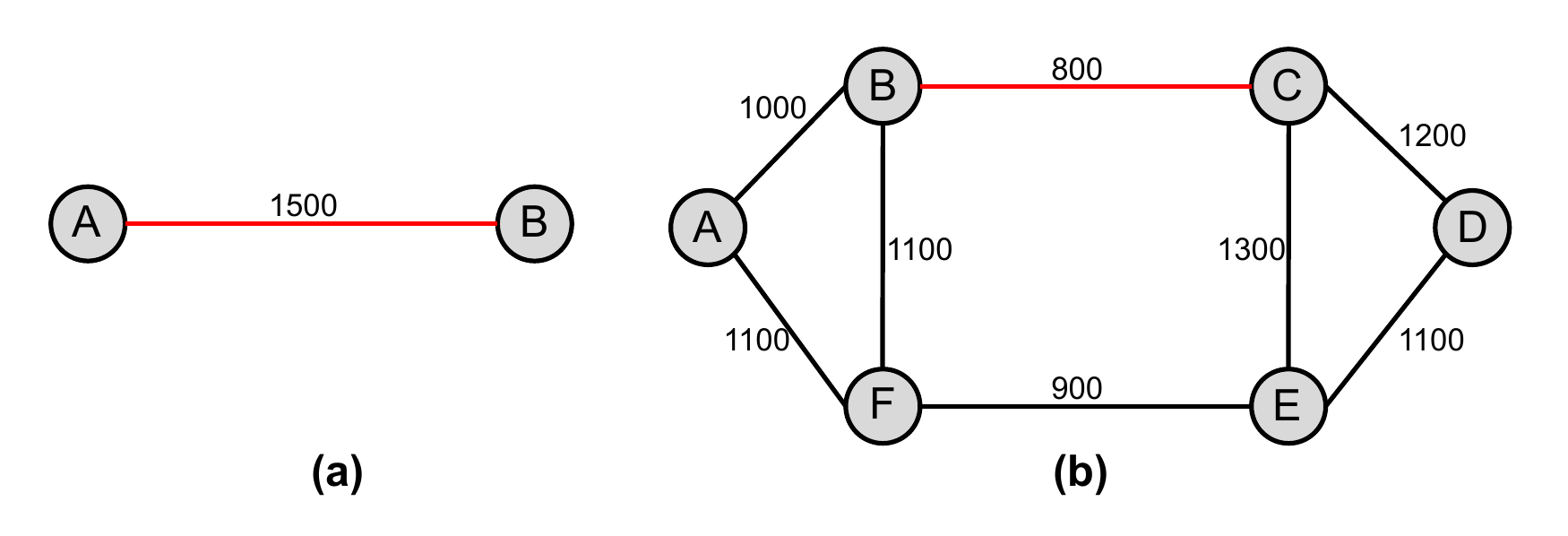}
        \caption{(a) Single-link and (b) evaluated topology. Distances are shown in km and jammed links in red.}
        \label{fig:topologies}
    \end{center}
\end{figure}

Simulations are made with ONS simulator \cite{ONS}. In both scenarios, we consider links with 320 slots for each direction. We generate 100,000 circuit requests, which follow a Poisson distribution with mean holding time of 600 seconds, following a negative exponential distribution and uniformly-distributed among all nodes-pairs. We use 3 different values of bandwidth which are 50, 100 and 200 Gbps, uniformly selected. To solve the routing and spectrum assignment  (RSA) problem, we adopt the basic Dijkstra algorithm for routing, and the \textit{First Fit}  policy for  slot assignment. The guardband between two adjacent lightpaths is assumed to be 2 slots. The circuits use 16QAM modulation, with $\text{SNR}_{th}$ equal to 15 \cite{fontinele2017efficient}. The parameters for the SNR calculation are shown in table \ref{tab:physical} \cite{fontinele2017efficient}:

\begin{table}[!h]
\centering
\caption{Parameters for SNR calculation.}
\label{tab:physical}
\begin{tabular}{|c|c|}
\hline
\textbf{Variable}   & \textbf{Value}                \\ \hline
$P_{TX}$            & 0 dB                          \\ \hline
$\Delta f$          & 12.5 GHz                      \\ \hline
$\alpha$            & 0.2 dB/km                     \\ \hline
L                   & 100                           \\ \hline
$\gamma$            & $1.22 Wkm^{-1}$               \\ \hline
$\beta_2$           & $16 ps^2/km$                  \\ \hline
$v$                 & $1.93 \times 10^{14} Hz$      \\ \hline
$F$                 & $6\;dB$                        \\ \hline
\end{tabular}
\end{table}

The jamming is fixed in a 10-slots range, with indexes from 140 to 149. Thus, all the generated circuits that are using those slots will be affected by the jamming attack. We vary the jamming power from 0 to 5 dB, in a step of 0.5 dB.

 In both scenarios (one-link and full network), there is a higher blocking probability for a jamming power close to 1.75 dB, reaching a value of blocking 20.1\% higher than the scenario without jamming for the full network. This is because with this jamming power, it is still possible to establish multiple jammed circuits, which generate greater impact on the QoT of the circuits that will still be established on the network.

As the power of the inserted jamming increases, the interference provided by the jammed circuits in the neighboring slots also increases. According to the proposed model, the interference is more intense between spectrally nearby slots. When evaluating the SNR between one established jammed circuit and a new jammed circuit with great spectral proximity, the control plane does not approve the establishment of the new jammed circuit, as the high interference violates the $\text{SNR}_{th}$. Therefore, less circuits under jamming interference are established close to each other, which causes less interference in the network as a whole and consequently there is a reduction in the blocking probability.

\subsection{Jamming Effect on a Single Link}

For evaluating the single link scenario, i.e two nodes and one link with multiple spans (Fig.~\ref{fig:topologies}(a)), we adopt the proposed attack model   shown in equation (\ref{eq:snr}). In our simulation, we consider 10 replications and we adopt the average value of results, considering a tolerable standard deviation.  The load is set to 120 Erlang. The average slots utilization is measured during the simulation as well as the blocking probability. To calculate the average slots utilization, the status of all slots is checked when a new request arrives in the network. The occupied slots (by circuits or by guard band) are identified, and their respective counter is increased by one. At the end, the utilization percentage is verified considering the total number of circuit requests generated in the simulation.

 We analyze mainly two metrics, which are the blocking probability and average slots utilization. The blocking probability indicates the relation between the number of blocked circuits in the network and the number of circuits generated, and it shown in figure \ref{fig:fullNetworkBlocking} for the two studied scenarios (one-link and full network). Figure \ref{fig:oneLink} illustrates the slot utilization rate for the single link scenario under different jamming power values. We zoom-in on the critical slots affected by jamming to see clearly the degree of harmness in term of utilization percentage. The range of jammed slots is colored with red for the ease of reading the figure.  

\begin{figure*}
 \centering
   \includegraphics[scale=0.35]{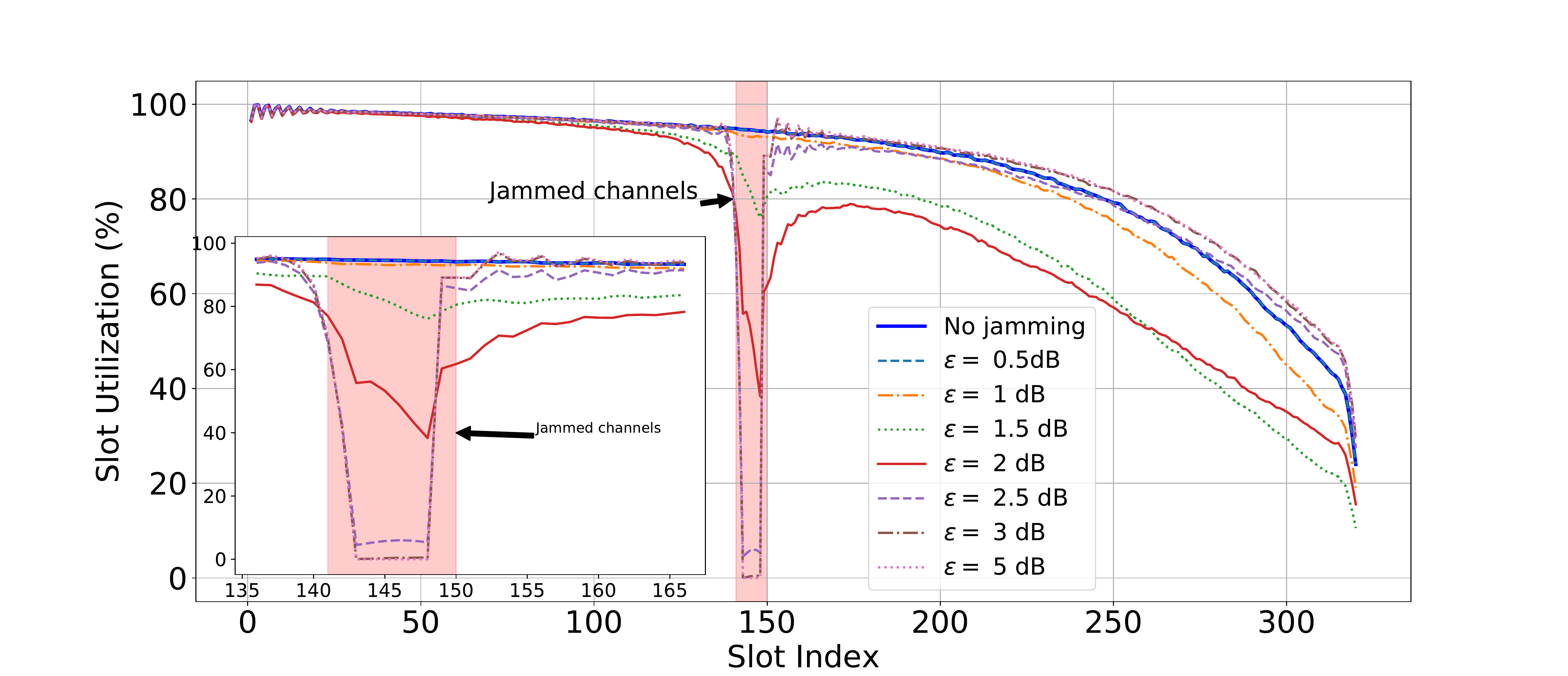}
 \caption{Jamming effect in a single link.}
\label{fig:oneLink}
\vspace{-0.2 cm}
\end{figure*}

The decreasing curve for the slot utilization in the scenario without the jamming attack is justified by the application of the First Fit policy, in which the allocation is preferably made in the lowest index slots. With the application of low power jamming (0.5 and 1.0 dB), the impact observed on the network remains indistinguishable, and the curve of slot utilization in Fig.~\ref{fig:oneLink} remains with behavior similar to the scenario without the jamming attack. However, the blocking probability increases exponentially in this interval of additional jamming values from 4\% to 5.8\% (see figure \ref{fig:fullNetworkBlocking}). From $1$ to $1.75$ dB jamming, slots utilization in the jamming interval and in neighboring slots, especially slots with higher indexes, start to decrease slowly, while the blocking probability increases considerably from 6\% to 19\%, to reach the max. value of blocking 19\% when  the additional jamming power is around 1.75 dB. (fig. \ref{fig:fullNetworkBlocking})

From the occurrence of 1.5 dB jamming, it is clear to see the impact of in-band jamming on the utilization of jammed slots (140-149) which is sharply reduced, until it reached the null value when the additional jamming power is higher than 3 dB.  Attacks with 1.5 dB and 2.0 dB jamming also have a huge impact on the slot utilization outside the jammed area, characterizing the out-of-band jamming. This is because the established jammed circuits have a low SNR value, near the $\text{SNR}_{th}$, and prevent the formation of new circuits, as these would take the SNR value of the circuits already established to values below the acceptable $\text{SNR}_{th}$ in the network. This impact is more intense on the slots near the jammed slot range. 

After 2.5 dB, jamming starts to affect slots in the attacked range more intensely. When a very intense jamming occurs, two jammed circuits cannot be allocated with great spectral proximity, as there is a greater impact on the QoT and consequently, on the SNR value. Thus, jamming values above 2.5 dB in the evaluated scenario turns unusable the jammed slot range, after the establishment of the first jammed circuit.

From $1.75$ to $3$ dB jamming,  the blocking probability decreases rapidly from 19\% to 4\%, which can be explained by the fact that the control plane tends to avoid the assignment of circuits in the jammed range due to their low SNR value. Thus, the jamming effect is rarely seen in the network and the blocking probability becomes lower when the utilization rate is close to zero.
\subsection{Jamming Effect on a Network}

In the full network scenario (Fig.~\ref{fig:topologies}(b)), 5 replications are simulated, with a load of 400 Erlang. The measured standard deviation values  are around 0.02 \%, and the average error value is around 0.01\% for the blocking probability. The jamming attack is applied in the link between nodes B and C, in a fixed slot range (140-149). \\Figure~\ref{fig:fullNetworkUtilization} shows the effects of jamming attack on slots utilization rate in the network, which represents the average utilization for all links in the network. We observe a less intense reduction in slots utilization compared to the link case. However, the reduction is still high in the range of jammed slots; for a jamming of 4 dB, the utilization of slots 140-149 is between 25\% and 30\%. When $\epsilon$ is higher than 3 dB, we remark a small oscillation of the utilization curve in the neighboring slots, then curves converge rapidly to the value of no jamming scenario. Thus, in a network, the out-of-band jamming effect is not visible for high jamming power when we monitor the average slots utilization. \\
Figure \ref{fig:fullNetworkBlocking} shows that the blocking probability curve for the network scenario has a similar monotonic behavior as the  single link scenario, but with lower blocking intensity. The maximum blocking probability is reached at a jamming power around 1,75 dB, but with an increase of only 2\% (from $\approx 7$\% to $\approx 9$\%).

\begin{figure*}
 \centering
   \includegraphics[scale=0.35]{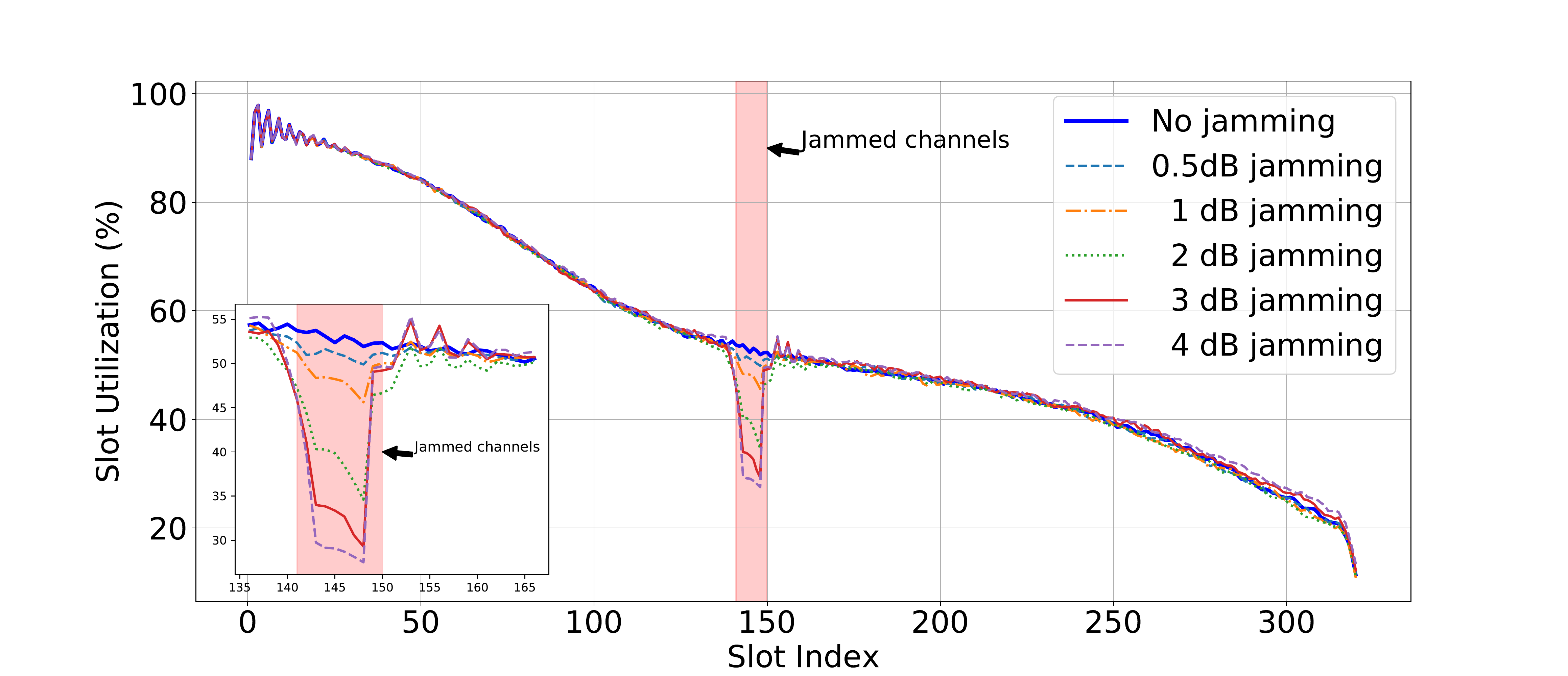}
 \caption{Jamming effect in a network.}
\label{fig:fullNetworkUtilization}
\vspace{-0.2 cm}
\end{figure*}

\begin{figure}
 \centering
   \includegraphics[scale=0.3]{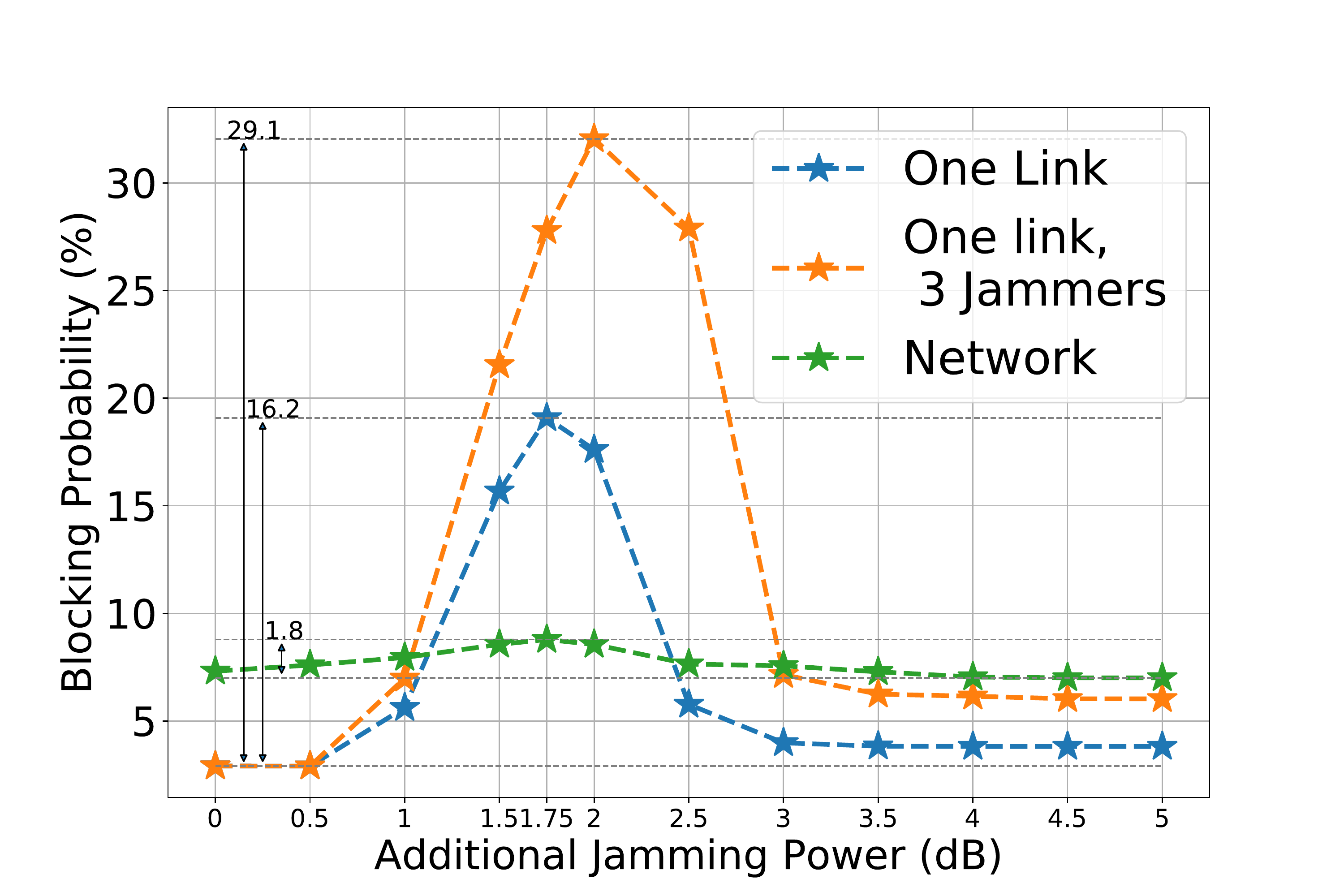}
 \caption{Blocking probability.}
\label{fig:fullNetworkBlocking}
\vspace{-0.2 cm}
\end{figure}

\subsection{Impact of Multiple Jammers}
Jamming attacks can be realized by more than one jammer to increase service disruption. As a result, it is interesting to consider a scenario in which the network is under several attacks. For the simplicity of the simulation, we adopt a network with a single link, we set the number of jammers to 3, where each jammer is using 10 consecutive slots: 50-59, 140-149, and 230-239.

\begin{figure*}
 \centering
   \includegraphics[scale=0.35]{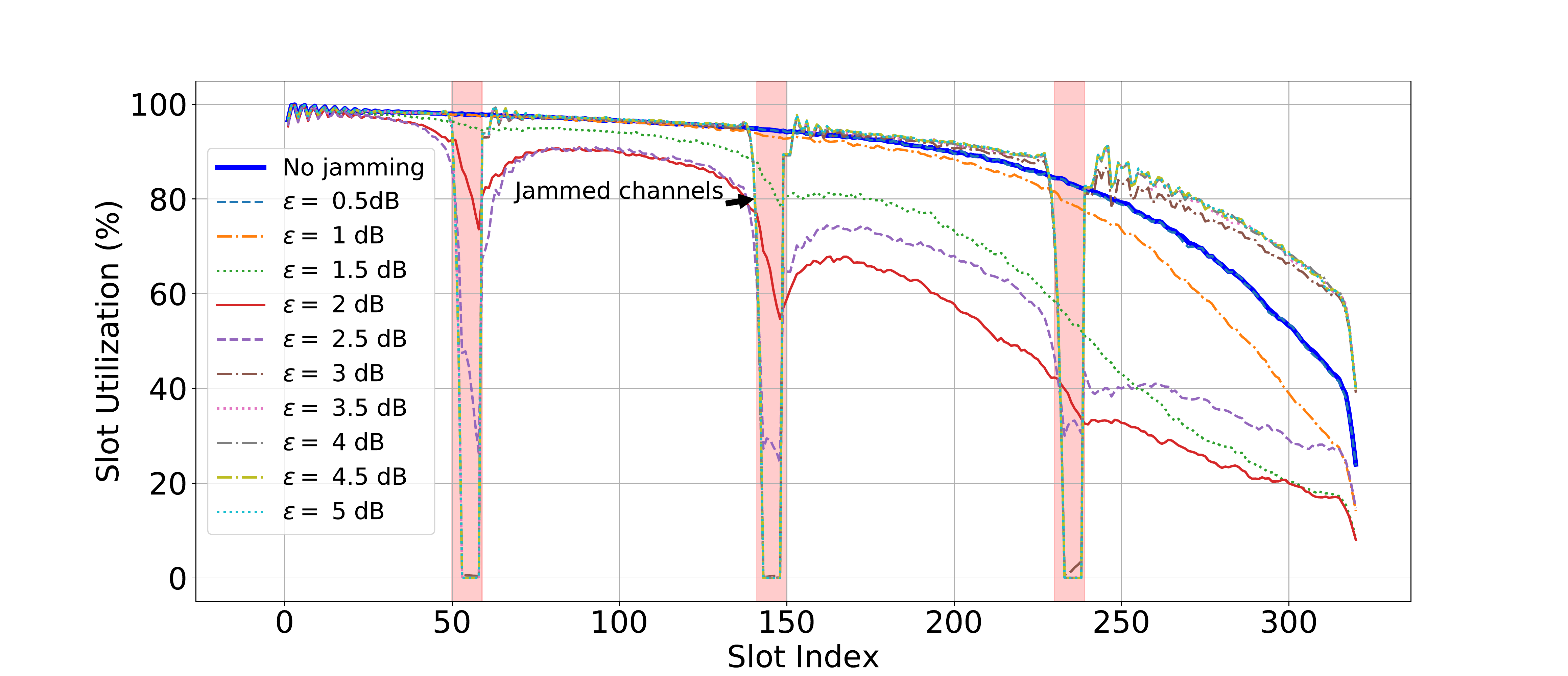}
 \caption{Impact of 3 jammers in a single link.}
\label{fig:onelinkmultiplejammers}
\vspace{-0.2 cm}
\end{figure*}

Figure \ref{fig:onelinkmultiplejammers} shows the impact of multi-jammers attack on slots utilization, regarding different values of jamming power, $\epsilon=0,0.5,..,5$ dB. When the jamming power is less than 1 dB, the effect on slots utilization is not clearly visible, while the blocking increase from 3\% to 7\%. When $\epsilon$ reaches the value 2 dB, we see the highest out-of-band effect where slots utilization  is affected with different intensities along all the link, and the blocking probability has its highest value around 32\%. For this value of jamming, the effect of in-band jamming is close to the effect of out-of-band jamming in terms of slots utilization, which means that this point presents the best value for an attacker to disrupt neighboring channels. In the previous parts, we saw that the highest blocking and highest out-of-band effect are obtained for a jamming power around 1.75 dB, and in this case when we added other jammers in different position, this \textit{best point} for jammer has changed. This result opens an interesting direction of research to study the impact of jamming position on blocking probability and slots utilization, and a theoretical analysis could be a good start for a future work.

\section{Conclusions} \label{sec:conclusion}

In this paper, we analyzed jamming power attacks in elastic optical networks, by embedding the insertion of jamming power into the physical layer model. We studied the impact of in-band and out-of-band jamming on the blocking probability and slots utilization into two different network topologies, considering different values of jamming power. Furthermore, we discussed the case of multiple attackers, where each one is using different frequency slots range. For in-band jamming, the results showed that slots utilization decreases with the increase of jamming power, and becomes null when the jamming power is higher than 3 dB, while for out-of-band jamming, the impact is maximal for a specific jamming power, 1.75 dB in our simulation. When we considered multiple positions of attackers, we obtained the highest blocking probability 32\% for a specific jamming power 2 dB, as a result, we conclude that the impact of jamming depends on attacker positions as well as the jamming power. As a future work, we plan to study the relationship between jamming power, attackers position, blocking probability and slots utilization, to provide a clear understanding of jamming attacks for detection and prevention studies.

\bibliographystyle{./IEEEtran}
\bibliography{mybib}

\end{document}